\documentclass[conference]{IEEEtran}
\IEEEoverridecommandlockouts

\usepackage{cite}
\usepackage{amsmath,amssymb,amsfonts}
\usepackage{graphicx}
\usepackage{textcomp}
\usepackage{xcolor}
\usepackage{amsthm}
\usepackage{algorithm}
\usepackage{hyperref}
\usepackage{bm}
\usepackage{algorithmicx,algpseudocode}
\usepackage{enumitem} \setlist[itemize]{leftmargin=*}
\usepackage{booktabs}
\usepackage{multirow}
\usepackage{tabularx}
\usepackage{array}

\newtheorem{proposition}{Proposition}

\newtheorem{lemma}{Lemma}

\DeclareMathOperator{\diag}{diag}

\DeclareMathOperator{\trace}{tr}

\def\BibTeX{{\rm B\kern-.05em{\sc i\kern-.025em b}\kern-.08em
    T\kern-.1667em\lower.7ex\hbox{E}\kern-.125emX}}
\begin{document}

\title{
Total Variation Sparse Bayesian Learning for Block Sparsity via Majorization–Minimization
}

\author{Yanbin He and Geethu Joseph\\
Signal Processing Systems Group, Delft University of Technology, The Netherlands\\Emails: \{y.he-1,g.joseph\}@tudelft.nl}

\maketitle

\begin{abstract}
Block sparsity is a widely exploited structure in sparse recovery, offering significant gains when signal blocks are known. Yet, practical signals often exhibit unknown block boundaries and isolated non-zero entries, which challenge traditional approaches. A promising method to handle such complex sparsity patterns is the difference-of-logs total variation (DoL-TV) regularized sparse Bayesian learning (SBL). However, due to the complex form of DoL-TV term, the resulting optimization problem is hard to solve. This paper develops a new optimization framework for the DoL-TV SBL cost function. By introducing an exponential reparameterization of the SBL hyperparameters, we reveal a novel structure that admits a majorization–minimization formulation and naturally extends to unknown noise variance estimation. Sparse recovery results on both synthetic data and extended source direction-of-arrival estimation demonstrate improved accuracy and runtime performance compared to benchmark methods.
\end{abstract}


\section{Introduction}\label{sec.intro}
Block sparsity, a signal structure in which most entries are zero and nonzeros occur in blocks, is widely studied in sparse recovery due to its practical relevance. For example, images in wavelets or DCT bases often produce groups of active coefficients \cite{charlie2021generating,zhang2024block}. In wireless communications, due to the clustered multipath components, channels can be block sparse in the angular domain \cite{hur2016proposal,he2023bayesian}. 
Consequently, various sparse recovery methods have been developed to exploit block sparsity, most of which are tailored for scenarios with known block boundaries.
Examples include adaptations of traditional compressed sensing algorithms such as basis pursuit~\cite{yuan2006model,van2009probing}, orthogonal matching pursuit~\cite{eldar2010block}, and sparse Bayesian learning (SBL)~\cite{zhang2011sparse}.
While these approaches can work well under known block boundaries, this assumption can be unrealistic for many practical applications \cite{zhang2024block}. Thus, our paper focuses on block sparse recovery with unknown block boundaries.

In practice, most block-sparse signals pose additional challenges, including unknown block boundaries and isolated nonzeros. For example, in the direction-of-arrival (DOA) estimation problem, the presence of both point and extended sources results in block sparse underlying signals with isolated nonzeros \cite{xenaki2016block}. Similarly, in automotive occupancy grid mapping, occupied regions form clusters whose sizes vary with obstacles and isolated nonzeros due to small obstacles or pedestrians~\cite{onen2024occupancy,zhai2025spatial}. To handle unknown boundaries, one approach is to use total variation (TV) regularization to adaptively promote blocks by discouraging rapid fluctuations \cite{liu2015image, aggarwal2016hyperspectral,nowak2011fused}. However, using TV directly on the signal can be inefficient and complicate the optimization problem~\cite{sant2022block}. 
An alternative is the pattern-coupled SBL (PCSBL) framework that operates on the hyperparameters that control the signal. It couples adjacent hyperparameters so that nearby entries are more likely to be active together, promoting blocks without manually setting block boundaries. However, it can bias the model towards block patterns, requires careful tuning, and degrades performance when the signal has a few isolated entries \cite{sant2022block}. This motivates a more flexible regularization that promotes block structures while tolerating isolated entries.

Recently, the difference of logs (DoL)-TV regularization with SBL has been shown to offer superior robustness and flexibility for block sparsity with few isolated entries, compared to other regularizers and benchmarks \cite{sant2022block}. However, since the convexity of the DoL-TV term relies on the relative magnitudes of adjacent hyperparameters, standard optimization techniques like majorization–minimization (MM) and convex formulations as in \cite{wipf2007new,wang2024iterative} cannot be directly used. 
Thus, \cite{sant2022block} relies on an expectation maximization-based method with alternating optimization between even and odd indices of hyperparameters, under the additional assumption of known noise variance. This assumption limits its practicality in real-world applications, where the noise variance is typically unknown. Thus, in this paper, we develop a novel optimization framework for the DoL-TV SBL that operates under unknown noise variance.

We present an MM framework for DoL-TV SBL that employs exponential reparameterization of its hyperparameters, enabling efficient optimization and noise variance estimation for real-world block-sparse signals with isolated nonzeros.
Our contributions are twofold.

\begin{itemize}
    \item \emph{MM Framework With Noise Estimation}: We design a novel MM framework for DoL-TV SBL cost that is transformed into a sum of convex and nonconvex terms. The DoL-TV term becomes convex, while the remaining nonconvex term can be efficiently majorized. This approach generates a sequence of convex subproblems and naturally supports joint estimation of the noise variance.
    
    \item \emph{Numerical Results}: We conduct numerical experiments on synthetic signals and extended source DOA estimation, validating the performance of our approach.
\end{itemize}

\section{DoL-TV SBL Cost}
In this section, we briefly introduce the block sparse recovery problem and the DoL-TV term following the formulation in \cite{sant2022block}. We consider the general multiple measurement vectors linear inversion problem as
\begin{equation}\label{eq.problem_basic}
     \bm Y = \bm H \bm X + \bm N,
\end{equation}
where $\bm Y \in \mathbb{C}^{M\times L}$ contains the measurements with $L$ snapshots, $\bm H \in \mathbb{C}^{M \times N}$ is the measurement matrix, $\bm X \in \mathbb{C}^{N\times L}$ is the unknown sparse matrix, and $\bm N$ is the noise matrix. We denote the $l$th column of $\bm X$ and $\bm N$ as $\bm x_l$ and $\bm n_l$, respectively. Here, $\bm x_l$ is block sparse, sharing the same support for different snapshots, and $\bm n_l$ is assumed to follow the additive zero-mean complex Gaussian distribution with variance $\lambda$, i.e., $\bm n_l\sim\mathcal{CN}(\bm 0, \lambda \bm I_M)$. Our goal is to recover block sparse matrix $\bm X$ using $\bm Y$ and $\bm H$. With $L=1$, our problem reduces to the standard block sparse recovery problem.

To impose the sparsity, SBL adopts a Gaussian prior $\mathcal{CN}(\bm 0, \bm\Gamma)$ on each $\bm x_l$, with $\bm\Gamma = \text{diag}(\bm\gamma)$ for $\bm \gamma\in\mathbb{R}^N$ for $l=1,\cdots,L$. Further, its hierarchical probabilistic model assumes an Inverse-Gamma distribution over $\bm \gamma$ and $\lambda$, 
\begin{align*}
    p(\bm \gamma) &= \prod_{n=1}^N \frac{b_0^{a_0}}{\Gamma(a_0)}\gamma_n^{-a_0-1}\exp(-\frac{b_0}{\gamma_n})\notag\\
    p(\lambda) &= \frac{b_0^{a_0}}{\Gamma(a_0)}\lambda^{-a_0-1}\exp(-\frac{b_0}{\lambda}),
\end{align*}
 where $\gamma_n$ is the $n$th element of $\bm \gamma$, $\Gamma(\cdot)$ is the gamma function, $a_0$ is the shape parameter, and $b_0$ is the scale parameter. 
To promote the block sparsity of each $\bm x_l$, we adopt the DoL-TV regularizer $\tau\sum_{n=2}^N |\log\gamma_n-\log\gamma_{n-1}|$, where $\tau\geq 0$ is the regularization weight. Estimating $\bm X$ using SBL then reduces to estimating $\bm\gamma$ and $\lambda$ through minimizing the equivalent objective function as
\begin{multline}
    \mathcal{L}(\bm\gamma,\lambda) = L\log |\bm \Sigma_{\bm Y}| + \trace(\bm Y^\mathsf{H} \bm \Sigma_{\bm Y}^{-1} \bm Y) + \tau\|\bm D\log \bm \gamma\|_1 \\+ \sum_{n=1}^N(a_0+1)\log \gamma_n+ \frac{b_0}{\gamma_n}+(a_0+1)\log\lambda+\frac{b_0}{\lambda}, 
    \label{eq:original_cost}
\end{multline}
where $\bm \Sigma_{\bm Y} = \lambda \bm I + \bm H \bm \Gamma \bm H^\mathsf{H}$, 
and $\bm D\in\{-1,0,1\}^{N-1\times N}$ is a differential matrix, $\|\bm D\log \bm \gamma\|_1 = \sum_{n=2}^N\left| \log \gamma_n - \log \gamma_{n-1} \right|$, forming the DoL-TV term as introduced in \cite{sant2022block}. 

DoL-TV SBL estimates the sparse vectors in $\bm X$ through solving~\eqref{eq:original_cost}, which is non-convex in $\bm\gamma$ and $\lambda$. The convexity of the DoL-TV term depends on the relative magnitudes of adjacent hyperparameters, preventing the construction of a universal majorizer~\cite{sant2022block} and making optimization difficult. In the next section, we present our alternative formulation where~\eqref{eq:original_cost} can be readily solved via iterative optimization.

\section{An Iterative Algorithm for DoL-TV SBL}
Inspired by the geodesic convexity in \cite{wiesel2012geodesic}, we reparameterize $\bm \gamma$ and $\lambda$ with variable $\bm z$ and $\beta$ such that $\bm \gamma = e^{\bm z}$ and $\lambda = e^{\beta}$, respectively \cite{boyd2004convex}. Here, operator $\exp(\cdot)$ is element-wise if the argument is a vector. The reparameterized cost is
\begin{multline}
    \mathcal{L}(\bm z,\beta) = L\log |\bm \Sigma_{\bm Y}| + \trace(\bm Y^\mathsf{H} \bm \Sigma_{\bm Y}^{-1} \bm Y) \\+ \mathcal{R}(\bm z,\beta)+ \tau\|\bm D \bm z\|_1,
    \label{eq:transformed_cost}
\end{multline}
with $\mathcal{R}(\bm z,\beta)=\sum_{n=1}^N(a_0+1)z_n+ b_0 e^{-z_n}+(a_0+1)\beta+{b_0}e^{-\beta}$ and $\bm \Sigma_{\bm Y} = e^{\beta} \bm I + \bm H \diag(e^{\bm z}) \bm H^\mathsf{H}$. 

This reparameterization reveals a structure that is not apparent in the original formulation and enables an MM-based optimization strategy. At a high level, all terms in the transformed cost except $\trace(\bm Y^\mathsf{H} \bm \Sigma_{\bm Y}^{-1} \bm Y)$ are convex in $(\bm z,\beta)$, while the nonconvex term can be efficiently majorized. Specifically, 
it is easy to verify that $\mathcal{R}(\bm z,\beta)$ and DoL-TV term $\tau\|\bm D \bm z\|_1$ are convex functions of $(\bm z,\beta)$. We now show that $L\log |\bm \Sigma_{\bm Y}|$ is also convex in $(\bm z,\beta)$ through Proposition \ref{prop.log_cvx}. 
\begin{proposition}\label{prop.log_cvx}
    Function $\log |e^{\beta} \bm I + \bm H \diag(e^{\bm z}) \bm H^\mathsf{H}|$ is convex in $(\bm z,\beta)\in \mathbb{R}^{N+1}$.
\end{proposition}

\begin{proof}
    See Appendix.
\end{proof}
The remaining term in the cost, i.e., $\trace(\bm Y^\mathsf{H} \bm \Sigma_{\bm Y}^{-1} \bm Y)$, is nonconvex, but it can be majorized via an upper bound through a variable $\bm \Theta\in\mathbb{C}^{N\times L}$ as~\cite{tipping2001sparse}
\begin{equation}\label{eq.upper_bound_quadra}
    \trace(\bm Y^\mathsf{H} \bm \Sigma_{\bm Y}^{-1} \bm Y) \leq e^{-\beta} \|\bm Y - \bm H \bm \Theta\|_\mathrm{F}^2 + \trace(\bm \Theta^\mathsf{H} \diag(e^{-\bm z}) \bm \Theta),
\end{equation}
with equality holding only when
\begin{equation}\label{eq.step_major}
    \bm \Theta = e^{-\beta}\left( e^{-\beta} \bm H^\mathsf{H} \bm H + \diag(e^{-\bm z}) \right)^{-1}\bm H^\mathsf{H}\bm Y.
\end{equation}
Thus, using~\eqref{eq.upper_bound_quadra} and~\eqref{eq.step_major}, we derive an MM procedure to solve the minimization of~\eqref{eq:transformed_cost}, or its equivalent problem~\eqref{eq:original_cost}.

\subsection{An MM Procedure}
We recall that MM solves a nonconvex cost by iteratively optimizing a convex upper-bound or surrogate function at each step. In the $k$th iteration, given the previous iterates $(\bm z^{(k-1)},\beta^{(k-1)})$, we first compute $\bm \Theta^{(k)}$ using~\eqref{eq.step_major}, which is then used to construct a surrogate function $g(\bm z,\beta|\bm \Theta^{(k)})$ as
\begin{multline}\label{eq.cost_mini}
    g(\bm z,\beta|\bm \Theta^{(k)})=L\log |\bm \Sigma_{\bm Y}| + e^{-\beta} \|\bm Y - \bm H \bm \Theta^{(k)}\|_\mathrm{F}^2 \\ + \trace((\bm \Theta^{(k)})^\mathsf{H} \diag(e^{-\bm z}) \bm \Theta^{(k)})+ \mathcal{R}(\bm z,\beta)+ \tau\|\bm D \bm z\|_1.
\end{multline}
From~\eqref{eq.upper_bound_quadra} and~\eqref{eq.step_major}, we see that this function majorizes $\mathcal{L}(\bm z,\beta)$ at $\bm \Theta^{(k)}$~\cite{sun2016majorization}, i.e., $\forall (\bm z,\beta)$,
\begin{align}
    g(\bm z,\beta|\bm \Theta^{(k)})&\geq \mathcal{L}(\bm z,\beta),\notag\\
    g(\bm z^{(k-1)},\beta^{(k-1)}|\bm \Theta^{(k)})&= \mathcal{L}(\bm z^{(k-1)},\beta^{(k-1)}),\notag
\end{align}
where the equality holds because $\bm \Theta^{(k)}$ minimizes~\eqref{eq.upper_bound_quadra} evaluated at $(\bm z^{(k-1)},\beta^{(k-1)})$. After constructing the surrogate function, the $k$th iteration solves the following problem
\begin{equation}\label{eq.step_min}
    (\bm z^{(k)},\beta^{(k)})=\arg\min_{\bm z,\beta} g(\bm z,\beta|\bm \Theta^{(k)}),
\end{equation}
where the cost function~\eqref{eq.cost_mini} is convex but not differentiable. We next present an alternating direction method of multipliers (ADMM)-based solution for \eqref{eq.step_min}. However, any convex solvers compatible with non-differentiable functions can be leveraged.

\subsection{ADMM Procedure}
ADMM is an iterative method that uses auxiliary variables to split the problem and dual variables to enforce constraints with coordinating updating. With an auxiliary variable $\bm u\in\mathbb{R}^{N-1}$ and a dual variable $\bm d\in\mathbb{R}^{N-1}$, the augmented Lagrangian of ADMM is 
\begin{multline}
    \tilde{\mathcal{L}} = L\log |\bm \Sigma_{\bm Y}| + e^{-\beta}\|\bm Y - \bm H \bm \Theta^{(k)}\|_\mathrm{F}^2  \\+ \trace((\bm \Theta^{(k)})^\mathsf{H} \diag(e^{-\bm z}) \bm \Theta^{(k)}) +\mathcal{R}(\bm z,\beta)\\+\tau\|\bm u\|_1+ \frac{\rho}{2}\|\bm D\bm z-\bm u + \bm d\|_2^2-\frac{\rho}{2}\|\bm d\|_2^2.\notag
\end{multline}
Here, $\bm u$ separates the non-differentiable DoL-TV term from $\bm z$ while $\bm d$ enforces the constraint $\bm u = \bm D \bm z$. ADMM iteratively updates $(\bm z, \beta)$, $\bm u$ and $\bm d$. So, the $t$th ADMM iteration has three steps as detailed below.


\subsubsection{Step 1}
We fix the auxiliary variable $\bm{u}=\bm u^{(k,t-1)}$ and the dual variable $\bm d=\bm d^{(k,t-1)}$, and optimize for $(\bm z, \beta)$ as
\begin{align}\label{eq.step1}
(\bm z^{(k,t)},\beta^{(k,t)})=&\arg\min_{\bm z,\beta}
    L\log |\bm \Sigma_{\bm Y}| + e^{-\beta}\|\bm Y - \bm H \bm \Theta^{(k)}\|_\mathrm{F}^2 \notag \\+& \trace((\bm \Theta^{(k)})^\mathsf{H} \diag(e^{-\bm z}) \bm \Theta^{(k)}) +\mathcal{R}(\bm z,\beta)\notag \\+& \frac{\rho}{2}\|\bm D\bm z-\bm u^{(k,t-1)} + \bm d^{(k,t-1)}\|_2^2.
\end{align}
Problem~\eqref{eq.step1} is convex and differentiable in $(\bm z, \beta)$, which can be solved using any off-the-shelf convex solver.

\subsubsection{Step 2}
We fix $(\bm z, \beta)=(\bm z^{(k,t)},\beta^{(k,t)})$ and the dual variable $\bm d=\bm d^{(k,t-1)}$, and update the auxiliary variable $\bm u$ as
\begin{align}\label{eq.step2}
    \bm u^{(k,t)}=&\arg\min_{\bm u}
    \tau\|\bm u\|_1+ \frac{\rho}{2}\|\bm D\bm z^{(k,t)} + \bm d^{(k,t-1)}-\bm u\|_2^2\notag\\
    =&\mathcal{S}_{\tau/\rho}(\bm D\bm z^{(k,t)} + \bm d^{(k,t-1)}).
\end{align}
Here, $\mathcal{S}_{\tau/\rho}(\cdot)$ is an element-wise soft thresholding operator, defined as $\mathcal{S}_{\tau/\rho}(\bm v)=\text{sgn}(\bm v)\odot\text{max}(|\bm v|-\tau/\rho,0)$, where $\text{sgn}(\cdot)$ returns the sign of each element of the argument vector $\bm v$ and $\odot$ is the Hadamard product.

\subsubsection{Step 3}
The last step is the dual update given by
\begin{equation}\label{eq.step3}
    \bm d^{(k,t)} = \bm d^{(k,t-1)} + \bm D\bm z^{(k,t)} - \bm u^{(k,t)},
\end{equation}
which concludes the $t$th iteration of ADMM. After $T$ iterations, the solution to~\eqref{eq.step_min} can be set as $(\bm z^{(k)},\beta^{(k)})=(\bm z^{(k,T)},\beta^{(k,T)})$. The resulting MM algorithm with ADMM is summarized in Algorithm \ref{al.exp_sbl_dol}. After $(\bm z,\beta)$ is obtained, the estimate of $\bm X$ is given by its MAP estimate, which is simply $\hat{\bm X}=\bm \Theta$ using~\eqref{eq.step_major}~\cite{sant2022block}. 

\begin{algorithm}[t]
\caption{Exponential-DoL SBL (Exp-DoL SBL)}
\label{al.exp_sbl_dol}
\begin{algorithmic}[1]
\Statex \textit {\bf Input:} Measurement $\bm Y\in\mathbb{C}^{M\times L}$, matrix $\bm H\in \mathbb{C}^{M\times N}$, maximum ADMM iteration number $T$, threshold $\epsilon$, parameters $\tau$ and $\rho$

\State \textbf{Initialization}: set $\bm z^{(0)}=\bm 0$, $\beta^{(0)}=0$, $\bm z^{(1)}=\bm 1$, $\beta^{(1)}=1$, $\bm u^{(0)}=\bm 1$, $\bm d^{(0)}=\bm 1$, and $k = 1$
\While{$\max\{\|\bm z^{(k)} - \bm z^{(k-1)}\|_\infty, |\beta^{(k)} - \beta^{(k-1)}|\}> \epsilon$}

\State Compute $\bm \Theta^{(k)}$ using~\eqref{eq.step_major} with $(\bm z^{(k-1)},\beta^{(k-1)})$

\State Initialize $\bm u^{(k,0)}=\bm u^{(k-1)}$, and $\bm d^{(k,0)}=\bm d^{(k-1)}$

\For{$t=1,2,\ldots, T$}

\State Compute $(\bm z^{(k,t)},\beta^{(k,t)})$ by solving~\eqref{eq.step1}

\State Compute $\bm u^{(k,t)}$ using~\eqref{eq.step2}

\State Compute $\bm d^{(k,t)}$ using~\eqref{eq.step3}


\EndFor

\State Set $(\bm z^{(k)},\beta^{(k)}) = (\bm z^{(k,T)},\beta^{(k,T)})$, $\bm u^{(k)}=\bm u^{(k,T)}$, and $\bm d^{(k)}=\bm d^{(k,T)}$

\State Set $k=k+1$

\EndWhile

\State Compute $\hat{\bm X}=\bm \Theta$ using~\eqref{eq.step_major} with $(\bm z^{(k-1)},\beta^{(k-1)})$

\Statex \textit {\bf Output:} Estimate $\hat{\bm X}$
\end{algorithmic}
\end{algorithm}




\section{Numerical Evaluation}

We present two sets of numerical results to evaluate our algorithm: synthetic signal recovery and extended sources DOA estimation. Benchmark algorithms are SBL \cite{wipf2004sparse}, PCSBL \cite{fang2014pattern}, EM DoL-TV SBL \cite{sant2022block}, and Adaptive-TV SBL \cite{djelouat2025adaptive}. We note that both EM DoL-TV SBL and Adaptive-TV SBL require noise variance as input, which we input the true noise variance. For a fair comparison, we present two versions of our Exp-DoL SBL, i.e., with and without the knowledge of noise variance by fixing $\beta$ at the ground truth and treating $\beta$ as variable, respectively. In Algorithm \ref{al.exp_sbl_dol}, we solve~\eqref{eq.step1} using $\texttt{fminunc}$ in MATLAB \cite{MATLAB} with $T=1$ ADMM iteration for efficiency.

\subsection{Synthetic Signals}\label{sec.eva_syn}

\begin{figure}[t]
    \centering
    \includegraphics[width=0.7\linewidth]{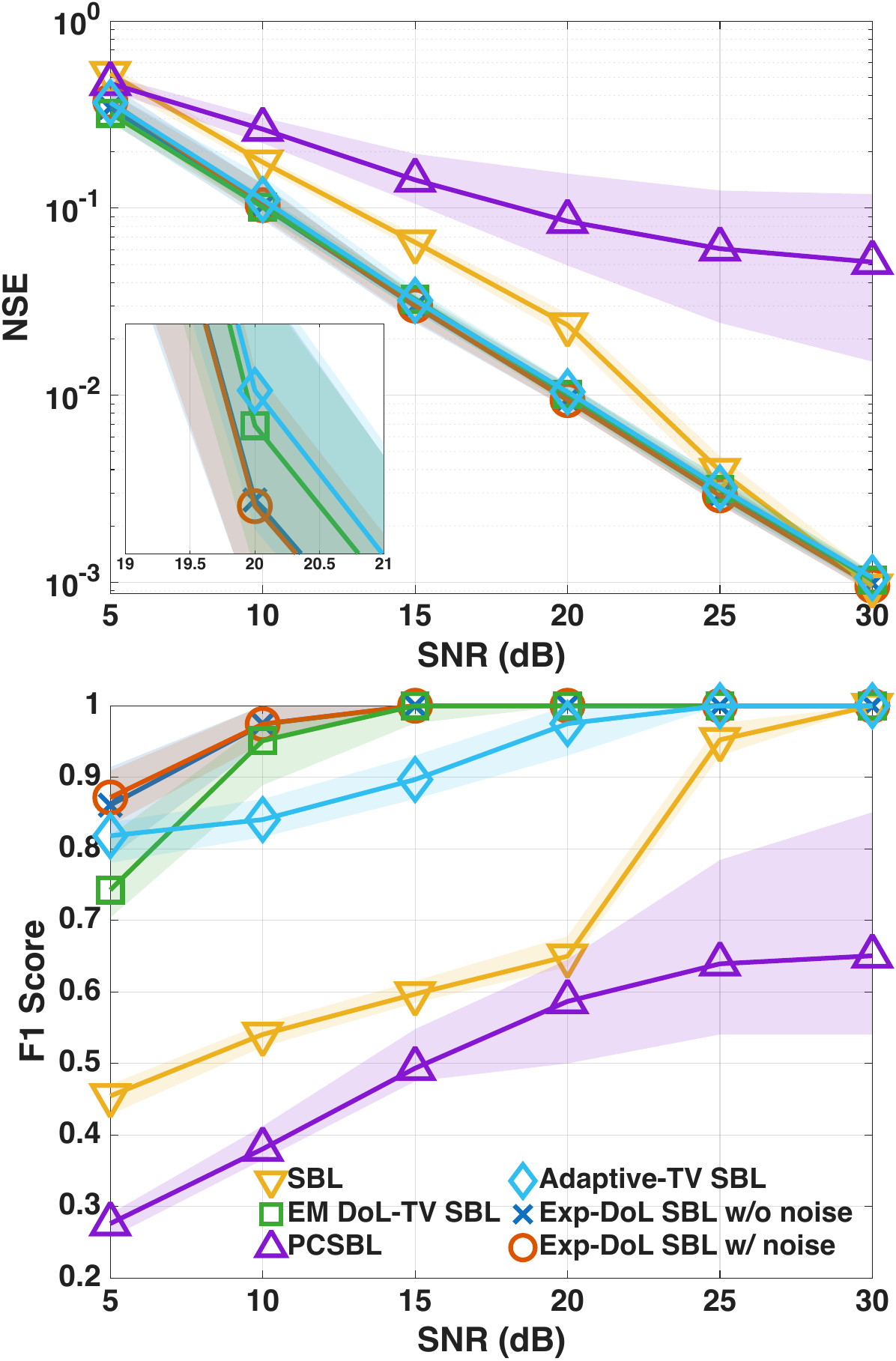}
    \caption{NSE and F1 score as functions of SNR.}
    \label{fig.nse_srr}
\end{figure}

\begin{figure*}
    \centering
    \includegraphics[width=1\linewidth]{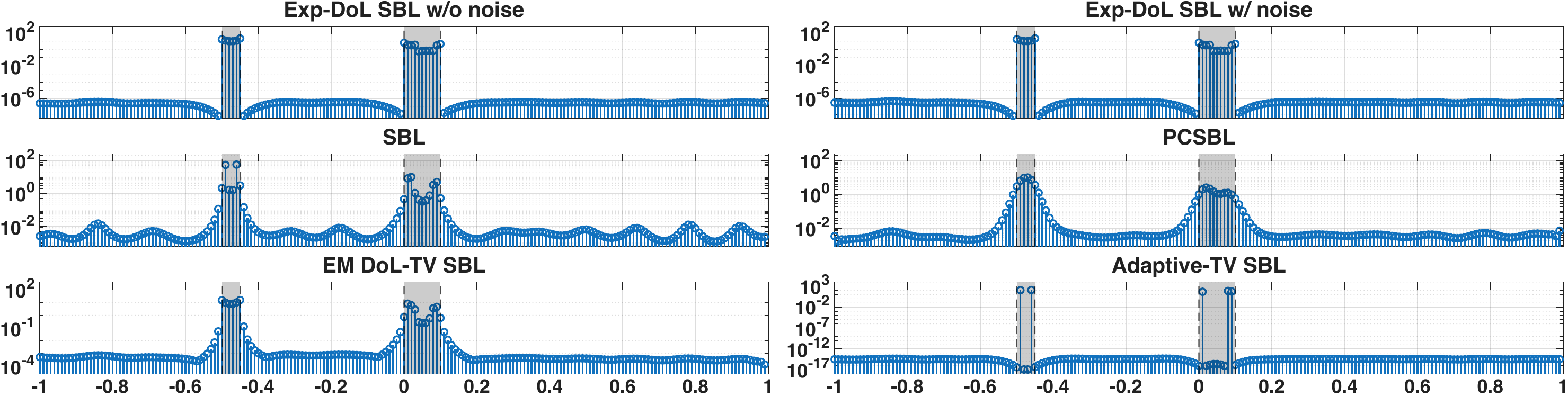}
    \caption{Power of reconstructed sparse signals across different angles. Shaded area: location of extended source.}
    \label{fig.doa}
\end{figure*}

We use~\eqref{eq.problem_basic} to generate synthetic data. We opt for $M=40$, $N=300$, and $L=5$. The entries of measurement matrix $\bm H$ are drawn independently from $\mathcal{CN}(0,1)$, with columns subsequently normalized to unit norm. We choose $\bm X$ \cite{sant2022block} to consist of three contiguous nonzero blocks of length five and five isolated nonzero entries with the nonzero entries drawn independently from $\mathcal{CN}(0,1)$. The noise variance $\lambda$ in~\eqref{eq.problem_basic} is determined by $\text{SNR~(dB)} = 10\log_{10}\mathbb{E}\{\|\bm H\bm X\|_\mathrm{F}^2/\|\bm N\|_\mathrm{F}^2\}$, with SNR values of $\{5,10,15,20,25,30\}$. For Algorithm \ref{al.exp_sbl_dol}, we empirically set $\tau=0.2$ and $\rho=1$.

We consider three metrics for performance evaluation: $25\%/50\%/75\%$ quartiles of normalized squared error (NSE), F1 score, and run time. We define $\text{NSE} = \|\bm X-\hat{\bm X}\|_\mathrm{F}^2/\|\bm X\|_\mathrm{F}^2$, where ${\bm X}$ is the ground truth and $\hat{\bm X}$ is its estimate. We follow the definition of F1 score in
\cite{sant2022block}, where $\text{F1}=1$ means perfect support recovery. We limit the number of EM/MM iterations for all algorithms to two hundred. The results are shown in Fig. \ref{fig.nse_srr} and Tab. \ref{tab.time_syn}.

Fig. \ref{fig.nse_srr} shows that PCSBL has the worst performance. This is due to its incapability to deal with isolated entries. Our Exp-DoL SBL has better NSE and F1 score performance than Adaptive-TV SBL and EM DoL-TV SBL, even when the noise variance is unknown. This is achieved with lower runtime except for $\text{SNR} = 5$~dB. Its improved accuracy and runtime demonstrate the effectiveness of our MM optimization procedure and its better adaptability to block sparse signal with isolated entries, especially in support recovery.

\begin{table}
    \centering
    \scriptsize
    \newcolumntype{Y}{>{\centering\arraybackslash}X}
    \setlength{\tabcolsep}{1.5pt} 
    \renewcommand{\arraystretch}{1.2}
    \caption{Average runtime in seconds. \textbf{Bold}: the best result.}
    \begin{tabularx}{0.48\textwidth}{c|YYYYYY}
        \hline
        SNR (dB) 
        & Exp-DoL SBL (w/o noise) 
        & Exp-DoL SBL (w/ noise) 
        & SBL 
        & PCSBL 
        & EM DoL-TV SBL 
        & Adaptive-TV SBL \\
        \hline
        5  & 1.0154 & 0.7792 & 2.0587 & 2.1834 & \textbf{0.6087} & 4.4844 \\
        10 & \textbf{0.5570} & 0.5635 & 1.5928 & 2.2012 & 0.5862 & 4.5207 \\
        15 & 0.4364 & \textbf{0.3924} & 1.1341 & 2.1424 & 0.6364 & 4.5988 \\
        20 & \textbf{0.3649} & 0.3864 & 0.9576 & 2.1704 & 0.6038 & 4.4447 \\
        25 & 0.3150 & \textbf{0.3008} & 0.8755 & 2.1909 & 0.5971 & 4.1597 \\
        30 & 0.3931 & \textbf{0.2950} & 0.9106 & 2.2129 & 0.5907 & 3.7283 \\
        \hline
    \end{tabularx}
    \label{tab.time_syn}
\end{table}

\subsection{Extended Sources DOA Estimation}

Following \cite{xenaki2016block}, we consider DOA estimation using half-wavelength spacing uniform linear array consisting of $M = 20$ elements. There are two spatially extended sources, covering the DOA intervals $[-0.5,-0.45]$ and $[0,0.1]$, respectively. Each continuous source is modeled with deterministic amplitude but uniformly distributed phase. The amplitudes of two sources are $1$ and $0.5$, respectively. The dictionary matrix is constructed based on steering vectors evaluating over a uniform grid in the spatial domain $[-1, 1)$ with a grid interval of $0.01$, resulting in $N = 200$ grid points. The number of snapshots is $L = 40$, and the received signals are corrupted by additive Gaussian noise with SNR fixed at $15$~dB.

Fig. \ref{fig.doa} illustrates the energy of the recovered sparse signals across the angular domain. Shaded areas denote true spatial spread of the extended sources. Methods such as SBL, PCSBL, and EM DoL-TV SBL suffer from high energy leakage, failing to clearly identify the source boundaries. Adaptive-TV SBL sufficiently suppresses energy leakage but fails to recover the continuous power distribution of the extended source, reducing ranges to a few isolated peaks. Our Exp-DoL SBL accurately reconstructs all the angles of the extended sources while maintaining low leakage, demonstrating its superiority.

\section{Conclusion}

We presented a novel optimization framework for the DoL-TV SBL cost to recover block sparse vectors with unknown block boundaries and isolated nonzero entries. By exponential reparameterization, we demonstrated that the DoL-TV SBL cost function admits an MM formulation, leading to a sequence of convex subproblems. Our method naturally enables noise variance estimation and improves practical applicability. Numerical results on recovery of synthetic data and extended source DOA estimation demonstrate consistent performance gains over benchmark methods. Future work includes reducing computational complexity, theoretical analysis of our framework, and the study of other regularizers.

\appendix

\section{Proof of Proposition \ref{prop.log_cvx}}

We need the following lemma to prove the result.
\begin{lemma}[Complex extension of Lemma 4 in \cite{wiesel2012geodesic}]\label{lm.complex_cvx}
    Let $\bm h_i \in \mathbb{C}^M$ for $i=1,\cdots,I$ be a set of vectors which span $\mathbb{C}^M$. The function $\log | \sum_{i=1}^I \exp(z_i)\bm h_i \bm h_i^\mathsf{H} |$ is convex in $\bm z\in \mathbb{R}^I$.
\end{lemma}

\begin{proof}
The convexity result for real-valued vectors is established in \cite[Lemma 4]{wiesel2012geodesic}, which we extend to the complex case. Accordingly, we define a mapping for a complex matrix $\bm h_i \bm h_i^\mathsf{H}\in\mathbb{C}^{M\times M}$ with $\bm h_i=\bm u_i+j\bm v_i$ as
\begin{align}
\mathcal{M}(\bm h_i \bm h_i^\mathsf{H})&=
\begin{bmatrix}
\Re(\bm h_i \bm h_i^\mathsf{H}) & -\Im(\bm h_i \bm h_i^\mathsf{H})\\
\Im(\bm h_i \bm h_i^\mathsf{H}) & \Re(\bm h_i \bm h_i^\mathsf{H})
\end{bmatrix}\notag
\\&=\begin{bmatrix}
\bm u_i\\
\bm v_i
\end{bmatrix}
\begin{bmatrix}
\bm u_i\\
\bm v_i
\end{bmatrix}^\mathsf{T} +
\begin{bmatrix}
-\bm v_i\\
\bm u_i
\end{bmatrix}
\begin{bmatrix}
-\bm v_i\\
\bm u_i
\end{bmatrix}^\mathsf{T}
\in\mathbb{R}^{2M\times 2M},\notag
\end{align}
where $\Re$ and $\Im$ denote the real and imaginary parts, respectively. Therefore, we derive
\begin{multline*}
    \mathcal{M}\left(\sum_{i=1}^I \exp(z_i)\bm h_i \bm h_i^\mathsf{H}\right) = \sum_{i=1}^I \exp(z_i)\mathcal{M}(\bm h_i \bm h_i^\mathsf{H})\notag\\=\sum_{i=1}^I \exp(z_i)\left(\begin{bmatrix}
\bm u_i\\
\bm v_i
\end{bmatrix}
\begin{bmatrix}
\bm u_i\\
\bm v_i
\end{bmatrix}^\mathsf{T} +
\begin{bmatrix}
-\bm v_i\\
\bm u_i
\end{bmatrix}
\begin{bmatrix}
-\bm v_i\\
\bm u_i
\end{bmatrix}^\mathsf{T}\right).\notag
\end{multline*}
Since $\bm u_i$ and $\bm v_i$ are real vectors, by Lemma 4 in \cite{wiesel2012geodesic}, to prove the convexity of $\log |\mathcal{M}(\sum_{i=1}^I \exp(z_i)\bm h_i \bm h_i^\mathsf{H}) |$, it is enough to show that $\{[\bm u_i^\mathsf{T}, \bm v_i^\mathsf{T}]^\mathsf{T},[-\bm v_i^\mathsf{T}, \bm u_i^\mathsf{T}]^\mathsf{T}\}_{i=1}^I$ spans $\mathbb{R}^{2M}$. 

To this end, we consider any nonzero vector $\bm w = [\bm a^\mathsf{T}, \bm b^\mathsf{T}]^\mathsf{T}\in\mathbb{R}^{2M}$. We recall that $\{\bm h_i\}_{i=1}^I$ spans $\mathbb{C}^M$, and therefore, there exists $i$ such that $\bm h_i^\mathsf{H}(\bm a + j\bm b)\neq 0$, i.e., 
\begin{align*}
    0&\neq \bm h_i^\mathsf{H}(\bm a + j\bm b)=(\bm u_i-j\bm v_i)^\mathsf{T}(\bm a+j\bm b)\\&=\bm u_i^\mathsf{T}\bm a + \bm v_i^\mathsf{T}\bm b + j\left(-\bm v_i^\mathsf{T}\bm a + \bm u_i^\mathsf{T}\bm b \right)\\&= \bm w^\mathsf{T}[\bm u_i^\mathsf{T}, \bm v_i^\mathsf{T}]^\mathsf{T}+j\bm w^\mathsf{T}[-\bm v_i^\mathsf{T}, \bm u_i^\mathsf{T}]^\mathsf{T}.\notag
\end{align*}
Therefore, we conclude that 
$$ \bm w^\mathsf{T}[\bm u_i^\mathsf{T}, \bm v_i^\mathsf{T}]^\mathsf{T} \neq 0 \quad \text{or} \quad \bm w^\mathsf{T}[-\bm v_i^\mathsf{T}, \bm u_i^\mathsf{T}]^\mathsf{T}\neq 0.$$
Thus, $\{[\bm u_i^\mathsf{T}, \bm v_i^\mathsf{T}]^\mathsf{T},[-\bm v_i^\mathsf{T}, \bm u_i^\mathsf{T}]^\mathsf{T}\}_{i=1}^I$ span $\mathbb{R}^{2M}$, and $\log |\mathcal{M}(\sum_{i=1}^I \exp(z_i)\bm h_i \bm h_i^\mathsf{H}) |$ is convex in $\bm z$ \cite[Lemma 4]{wiesel2012geodesic}. 

The final step is to connect $\log|\mathcal{M}(\sum_{i=1}^I \exp(z_i)\bm h_i \bm h_i^\mathsf{H})|$ and $\log | \sum_{i=1}^I \exp(z_i)\bm h_i \bm h_i^\mathsf{H} |$. By \cite[Lemma 1]{telatar1999capacity}, for any Hermitian positive definite matrix $\bm U$, $\mathcal{M}(\bm U)$ is real positive definite with $\det(\mathcal{M}(\bm U)) = \det^2(\bm U)$. Thus, we have
\begin{equation}
     \log \left| \sum_{i=1}^I \exp(z_i)\bm h_i \bm h_i^\mathsf{H} \right|=\frac{1}{2}\log\left|\mathcal{M}\left(\sum_{i=1}^I \exp(z_i)\bm h_i \bm h_i^\mathsf{H}\right)\right|.\notag
\end{equation}
Hence, $\log | \sum_{i=1}^I \exp(z_i)\bm h_i \bm h_i^\mathsf{H} |$ is also convex in $\bm z$ since positive scaling preserves convexity.
\end{proof}

\subsection*{Proof of Proposition \ref{prop.log_cvx}}

We rewrite $e^{\beta} \bm I + \bm H \diag(e^{\bm z}) \bm H^\mathsf{H}$ with $\bm e_m$ as the $m$th column of the identity matrix as
\begin{equation*}
    e^{\beta} \bm I + \bm H \diag(e^{\bm z}) \bm H^\mathsf{H}= e^\beta\sum_{m=1}^M \bm e_m \bm e_m^\mathsf{H} + \sum_{n=1}^Ne^{z_n}\bm h_n \bm h_n^\mathsf{H}.
\end{equation*}
Since $\{\bm e_m\}_{m=1}^M\cup\{ \bm h_n\}_{n=1}^N$ spans $\mathbb{C}^M$, according to Lemma \ref{lm.complex_cvx} and the fact that convexity is preserved under affine composition, we establish the convexity.

\bibliographystyle{ieeetr}
\bibliography{refs}

\end{document}